\documentclass{llncs} 
\usepackage{xspace}
\usepackage{xcolor}
\usepackage{mdframed}
\usepackage{amsmath}
\usepackage{amsfonts}
\usepackage{amssymb}
\usepackage{paralist}
\usepackage{turnstile}
\usepackage{microtype}
\usepackage[lined,ruled,vlined,linesnumbered]{algorithm2e}
\usepackage{mdwtab}
\usepackage{wrapfig}
\usepackage[english]{babel}
\usepackage{blindtext}
\usepackage{subcaption}

\usepackage{listings} 
\def\prog{\ensuremath{\mathsf{prog}}\xspace}

\def\translate{\ensuremath{\mathsf{translate}}\xspace}
\def\subprog{\ensuremath{\mathsf{subprog}}\xspace}
\def\RR{{\mathbb{R}}}

\def\RRules{\ensuremath{\mathsf{RR}}\xspace}

\newcommand{\rankpbe}{\ensuremath{\mathsf{rank}_{\mathrm{PBE}}}\xspace}
\newcommand{\rankc}{\ensuremath{\mathsf{rank}_{\mathrm{c}}}\xspace}
\newcommand{\sample}{\ensuremath{\mathsf{sample}}\xspace}

\newcommand{\eatme}[1]{}
\xspace

\newcommand{\ie}{\emph{i.e.}\xspace}
\newcommand{\flashfill}{{\fontsize{9}{10}\selectfont\textsf{FlashFill}}\xspace}
\newcommand{\flashprofile}{{\fontsize{9}{10}\selectfont\textsf{FlashProfile}}\xspace}

\newcommand{\flashmeta}{{\small\textsf{FlashMeta}}\xspace}
\newcommand{\secref}[1]{Section~\ref{section:#1}\xspace}
\newcommand{\subsecref}[1]{Section~\ref{subsection:#1}\xspace}
\newcommand{\figref}[1]{Figure~\ref{figure:#1}\xspace}
\newcommand{\algoref}[1]{Algorithm~\ref{algorithm:#1}\xspace}

\newcommand{\lang}{\ensuremath{\mathsf{PL}}\xspace}
\newcommand{\dsllang}{\ensuremath{{L}}\xspace}

\newcommand{\esmall}{\ensuremath{E}\xspace}

\newcommand{\specification}[1]{\ensuremath{\phi_{#1}}\xspace}

\newcommand{\qsynth}{\textsc{qpbe}\xspace}

\newcommand{\gtype}[1]{\texttt{#1}\xspace}
\newcommand{\gsymbol}[1]{\emph{#1}\xspace}
\newcommand{\expandsto}{\texttt{:=}\xspace}
\newcommand{\goperator}[1]{\textsc{#1}\xspace}

\newcommand{\boldpara}[1]{\vspace*{0.2ex}\noindent\textbf{#1}\xspace}
\newcommand{\itpara}[1]{\vspace*{0.2ex}\noindent\emph{#1}\xspace}
\newcommand{\rankingscheme}{\ensuremath{\mathsf{rank}}\xspace}

\newcommand{\fullinputset}{\ensuremath{\mathsf{pre}}\xspace}
\newcommand{\qpbe}{{qPBE}\xspace}

\newcommand{\ignore}[1]{}

\newcommand{\phaseone}{optimization phase\xspace}

\newcommand{\phasetwo}{rewrite phase\xspace}

\newcommand{\phasesoneandtwo}{optimization and rewrite phases\xspace}

\usepackage{tikz}
\usepackage{pgfplots}
\pgfplotsset{compat=1.12}
\usetikzlibrary{arrows}
\usetikzlibrary{arrows.meta}
\usetikzlibrary{automata}
\usetikzlibrary{positioning}
\usetikzlibrary{shapes}
\usetikzlibrary{shapes.geometric}
\usetikzlibrary{shapes.misc}
\usetikzlibrary{decorations}
\usetikzlibrary{shadows}
\usetikzlibrary{decorations.pathmorphing}
\usetikzlibrary{automata, arrows, positioning}
\usetikzlibrary{backgrounds,calc}
\usetikzlibrary{patterns}
\usetikzlibrary{trees}
\usetikzlibrary{fit}

\tikzstyle{block} = [rectangle, draw, fill=blue!20, text centered, rounded corners]
\tikzstyle{line} = [draw, -latex']
\tikzstyle{cloud} = [draw, ellipse,fill=red!20, node distance=3cm,
    minimum height=2em]

\tikzset{
  cascaded/.style = {
    general shadow = {
      shadow scale = 1,
      shadow xshift = -1ex,
      shadow yshift = 1ex,
      draw,
      thick,
      fill = white},
    general shadow = {
      shadow scale = 1,
      shadow xshift = -.5ex,
      shadow yshift = .5ex,
      draw,
      thick,
      fill = white},
    fill = white,
    draw,
    thick,
    minimum width = 3cm,
    minimum height = 2cm
  }
}

\tikzset{
  programmer figure/.style = {
    shape = programmershape,
    draw = black,
    very thick,
    fill = green!80,
    minimum size = 2cm
  },
  basic box/.style = {
    font = \footnotesize,
    shape = rectangle,
    align = center,
    draw  = black,
    fill  = #1!33,
    inner sep = 2mm,
    very thick,
    rounded corners},
  input box/.style = {
    font =\footnotesize,
    shape = rectangle,
    align = center,
    draw  = black,
    fill  = black!10,
    inner sep = 3mm,
    very thick},
  decision/.style = {
    font = \footnotesize,
    shape = diamond,
    align = center,
    draw  = black,
    fill  = #1!33,
    text badly centered,
    inner sep=2mm,
    very thick
  },
  header node/.style = {
    Minimum Width = header nodes,
    font          = \strut\Large\ttfamily,
    text depth    = +0pt,
    fill          = white,
    draw},
  header/.style = {
    inner ysep = +1.5em,
    append after command = {
      \pgfextra{\let\TikZlastnode\tikzlastnode}
      node [header node] (header-\TikZlastnode) at (\TikZlastnode.north) {#1}
      node [span = (\TikZlastnode) (header-\TikZlastnode)]
        at (fit bounding box) (h-\TikZlastnode) {}
    }
  },
  hv/.style = {to path = {-|(\tikztotarget)\tikztonodes}},
  vh/.style = {to path = {|-(\tikztotarget)\tikztonodes}},
  fat blue line/.style = {ultra thick, blue},
  data line/.style = {ultra thick, black, densely dashed},
  black line/.style = {ultra thick, black},
  thick flow line/.style = {black,line width=1.33mm,->, >=stealth},
  comm line/.style={very thick,black, ->, >=stealth},
  bidir comm line/.style={very thick,black, <->, >=stealth}
}

\tikzset{>=stealth', state/.append style={minimum size=.4cm, align=center}}
\definecolor{dark-green}{rgb}{0,.6,0}
\tikzstyle{transition}=[->, >=stealth',draw,thick]
\tikzstyle{execution state}=[rectangle, draw, text centered, rounded corners, align=center]
\tikzstyle{transition label}=[above, align=center]
\tikzset{snake it/.style={thick, decorate, decoration={snake,amplitude=1pt,segment length=0.3cm}}}

\tikzstyle{lane} = [draw, loosely dotted, ->, >=latex, line width=1.3]
\tikzstyle{message} = [draw, ->, >=latex, line width=1.5]
\tikzstyle{halfmsg} = [draw, -{Rays[]}, >=latex, line width=1.5]
\tikzstyle{intmsg} = [draw, ->, >=latex, line width=1.5]
\tikzstyle{tentmessage} = [draw, dashed, ->, >=latex, line width=1.5]
\tikzstyle{message node} = [sloped, align=center, above, pos=.5, font=\small]
\tikzstyle{command node} = [align=center,font=\small,]
\tikzstyle{lane label} = [align=left, font=\small]
\tikzstyle{figure letter} = [font={\LARGE\bf}, below=.2cm]

\newcounter{sarrow}

\makeatletter
\pgfdeclareshape{programmershape}
{
  \inheritsavedanchors[from=circle] 
  \inheritanchorborder[from=circle]
  \inheritanchor[from=circle]{north}
  \inheritanchor[from=circle]{north west}
  \inheritanchor[from=circle]{north east}
  \inheritanchor[from=circle]{center}
  \inheritanchor[from=circle]{west}
  \inheritanchor[from=circle]{east}
  \inheritanchor[from=circle]{mid}
  \inheritanchor[from=circle]{mid west}
  \inheritanchor[from=circle]{mid east}
  \inheritanchor[from=circle]{base}
  \inheritanchor[from=circle]{base west}
  \inheritanchor[from=circle]{base east}
  \inheritanchor[from=circle]{south}
  \inheritanchor[from=circle]{south west}
  \inheritanchor[from=circle]{south east}

  \backgroundpath{
    \def\progratio{1.0}
    \def\bodyratio{0.7}
    \centerpoint
    \pgfutil@tempdima=\radius
    \pgfutil@tempdima=\radius
    \pgfmathsetlength{\pgfutil@tempdimb}{sqrt(4 / ((\progratio * \progratio) + 1)) * \pgfutil@tempdima}
    \pgfmathsetlength{\pgf@xc}{sqrt((\pgfutil@tempdima * \pgfutil@tempdima) - (0.25 * \pgfutil@tempdimb * \pgfutil@tempdimb))}

    \pgfmathsetlength{\pgf@yc}{\bodyratio * \progratio * \pgfutil@tempdimb}
    \centerpoint
    \pgf@xa=\pgf@x \pgf@ya=\pgf@y
    \advance\pgf@xa by -0.5\pgfutil@tempdimb
    \advance\pgf@ya by -\pgf@xc

    \pgfpathmoveto{\pgfpoint{\pgf@xa}{\pgf@ya}}

    \advance\pgf@ya by \pgf@yc
    \advance\pgf@xa by 0.1\pgfutil@tempdimb
    \pgfpathlineto{\pgfpoint{\pgf@xa}{\pgf@ya}}

    \advance\pgf@xa by 0.8\pgfutil@tempdimb
    \pgfpathlineto{\pgfpoint{\pgf@xa}{\pgf@ya}}

    \advance\pgf@ya by -\pgf@yc
    \advance\pgf@xa by 0.1\pgfutil@tempdimb
    \pgfpathlineto{\pgfpoint{\pgf@xa}{\pgf@ya}}
    \pgfpathclose
  }

  \foregroundpath{
    \def\progratio{1.0}
    \def\bodyratio{0.6}
    \centerpoint
    \pgf@xa=\pgf@x
    \pgf@ya=\pgf@y
    \pgfutil@tempdima=\radius
    \pgfmathsetlength{\pgfutil@tempdimb}{sqrt(4 / ((\progratio * \progratio) + 1)) * \pgfutil@tempdima}

    \advance\pgf@ya by 0.33\pgfutil@tempdimb
    \pgfpathcircle{\pgfpoint{\pgf@xa}{\pgf@ya}}{0.4\pgfutil@tempdima}
    \pgfpathclose
    \pgfsetfillcolor{rgb:orange,1;yellow,2;pink,5}
    \pgfusepath{stroke,fill}

    \centerpoint
    \pgf@xa=\pgf@x
    \pgf@ya=\pgf@y
    \advance\pgf@ya by 0.4\pgfutil@tempdimb
    \advance\pgf@xa by -0.36\pgfutil@tempdima

    \pgfpathmoveto{\pgfpoint{\pgf@xa}{\pgf@ya}}
    \pgfpatharc{-180}{0}{0.09\pgfutil@tempdima}
    \pgfpatharc{-180}{0}{0.09\pgfutil@tempdima}
    \pgfpatharc{-180}{0}{0.09\pgfutil@tempdima}
    \pgfpatharc{-180}{0}{0.09\pgfutil@tempdima}
    \pgfpatharc{13}{167}{0.38\pgfutil@tempdima}
    \pgfpathclose
    \pgfsetfillcolor{black}
    \pgfusepath{stroke,fill}
  }
}

\pgfdeclareshape{scriptbox}{
  \inheritsavedanchors[from=rectangle] 
  \inheritanchorborder[from=rectangle]
  \inheritanchor[from=rectangle]{north}
  \inheritanchor[from=rectangle]{north west}
  \inheritanchor[from=rectangle]{north east}
  \inheritanchor[from=rectangle]{center}
  \inheritanchor[from=rectangle]{west}
  \inheritanchor[from=rectangle]{east}
  \inheritanchor[from=rectangle]{mid}
  \inheritanchor[from=rectangle]{mid west}
  \inheritanchor[from=rectangle]{mid east}
  \inheritanchor[from=rectangle]{base}
  \inheritanchor[from=rectangle]{base west}
  \inheritanchor[from=rectangle]{base east}
  \inheritanchor[from=rectangle]{south}
  \inheritanchor[from=rectangle]{southwest}
  \inheritanchor[from=rectangle]{southeast}
  \backgroundpath{
    \southwest \pgf@xa=\pgf@x \pgf@ya=\pgf@y
    \northeast \pgf@xb=\pgf@x \pgf@yb=\pgf@y
    \pgfmathsetlength{\pgfutil@tempdima}{\pgf@xb - \pgf@xa}
    \pgfpathmoveto{\pgfpoint{\pgf@xa}{\pgf@ya}}
    \pgfpathlineto{\pgfpoint{\pgf@xa}{\pgf@yb}}
    \pgfpathlineto{\pgfpoint{\pgf@xb}{\pgf@yb}}
    \pgfpathlineto{\pgfpoint{\pgf@xb}{\pgf@ya}}
    \pgfpatharc{300}{240}{0.5\pgfutil@tempdima}
    \pgfpatharc{60}{120}{0.5\pgfutil@tempdima}
    \pgfpathclose
  }
}

\makeatother

\newsavebox\scenariobox
\newsavebox\automatonbox
\newsavebox\bufferbox

\begin{document}
\title{Quantitative Programming by Examples}
\author{Sumit Gulwani and Kunal Pathak and Arjun Radhakrishna and Ashish Tiwari and Abhishek Udupa}
\institute{Microsoft Corporation, Redmond, WA}
\maketitle
\vspace{-5ex}
\begin{abstract}
Programming-by-Example (PBE) systems synthesize an intended program in some (relatively constrained) domain-specific language from a small number of input-output examples provided by the user. 
In this paper, we motivate and define the problem of quantitative PBE (\qpbe) that relates to synthesizing an intended program over an underlying (real world) programming language that also minimizes a given quantitative cost function. 
We present a modular approach for solving \qpbe that consists of three phases: intent disambiguation,
global search, and local search. 
%
%
On two concrete objectives, namely program performance and size,
our \qpbe procedure achieves $1.53 X$ and $1.26 X$ improvement respectively over  the baseline \flashfill PBE system, averaged over $701$ benchmarks.
Our detailed experiments validate the design of our procedure and show 
value of combining global and local search for \qpbe. 
\end{abstract}

\section{Introduction}
\label{section:introduction}
Programming-by-Example (PBE) systems synthesize an intended program in an
underlying domain-specific language from a small number of input-output 
examples provided by the user~\cite{pbe:aplas17}. 
Various PBE systems have been successfully deployed in practice; e.g.,
the \flashfill feature in Microsoft Excel for performing string 
transformations~\cite{flashfill}, 
the Extract-from-Web feature in PowerBI for extracting tables from
web pages~\cite{raza-17}, and the ConvertFrom-String cmdlet in Powershell
for extracting tables from custom text files~\cite{flashextract:pldi14}. 
These systems are popular with {\em end users}, who want to automate their
one-off repetitive tasks on small amounts of data, where correctness
can be verified quickly by examining the output.

Unfortunately, the PBE formalism does not provide the user any control on the 
nature of the synthesized program. 
For example, {\em data scientists} who deal with large data would like to 
direct a PBE system to generate efficient programs; while {\em developers}
who would like to incorporate synthesized programs as a part of their source
code would prefer that PBE systems generate small/readable programs in a 
specific programming language.

We model these new requirements on synthesized programs as quantitative objectives
and define the quantitative PBE problem (\qpbe).
The \qpbe problem asks a synthesizer to not only synthesize the intended program
in a rich target language, but to also produce a program that optimizes a given
cost function.
Our solution methodology is to leverage existing PBE solvers, which are adept
at synthesizing the intended program from a small number of input-output examples
by leveraging a ranking function over programs, albeit in a relatively constrained
domain-specific language (DSL). 
There are two key challenges to take care of: 
how to account for the cost function that expresses the quantitative objective,
and how to leverage constructs from the rich target language that can further 
improve the quantitative objective.

Our approach for \qpbe uses three phases to solve the problem.
We first invoke the PBE solver to generate an intended program in the DSL from the
small number of input-output examples.
We use this intended program to generate a more complete specification in the form of
a larger set of examples.  The goal of the first phase is intent disambiguation using
a small number of examples.

Our first key idea is to replace the ranking function of the PBE solver by a  custom
cost function tuned towards optimizing the objective of \qpbe.
Hence, in the second phase, we re-invoke the PBE solver, but with the more comprehensive set of examples
(which avoids the need for the intent-based ranking function) and using the objective-based
ranking function.
This step performs a global search on the space of DSL programs, and
yields a correct program in the DSL that is also optimal with respect
to the cost function. 

Our second key idea is to bridge the gap between the constrained DSL and the desired target language
by means of rewrite rules that describe transformations from the DSL to expressions that have natural
translation to the target language.
These rewrite rules are not semantics-preserving in general, but only when the inputs satisfy some
preconditions. We apply the rewrites to the program generated in Phase $2$ only when 
they are sound and objective-decreasing. Thus, the third phase performs local search starting from the
program generated in Phase $2$ to get the final program in the target language.


We evaluate our 3-phase approach on synthesis of small (and hence readable) and efficient string 
manipulation programs in Python from input-output examples.
We leverage the \flashfill PBE synthesizer for this purpose, which operates over a 
constrained domain-specific language that includes operators like substring, concatenate,
case conversion, and date parsing.
We show that our methodology generates programs that are significantly smaller and more 
efficient than those produced by the stand alone invocation of the \flashfill synthesizer.
For the performance and program size objectives, respectively, we generate programs that
are $2-3$ orders of magnitude faster and $1$ order of magnitude smaller than the ones
produced by \flashfill.

\section{Motivating Example}
\label{section:motivating_example}

\begin{wrapfigure}[6]{R}{0.63\textwidth}
\vspace*{-9mm}
\begin{mdframed}
\fontsize{8}{9}\selectfont
~~~~~~~~~~~~~~~~Input~~~~~~~~~~~~~~~~~~~~~~~~~~~~~~Output\\
\texttt{"06/08/2010 and 08/05/2010"} $\mapsto$ \texttt{"August 5, 2010"}\\
\texttt{"04/02/2008 and 03/31/2010"} $\mapsto$ \texttt{"March 31, 2010"}\\
\texttt{"04/02/2008 and 06/22/2015"} $\mapsto$ \texttt{????}
. . .
\end{mdframed}
\vspace*{-4mm}
\caption{\fontsize{8}{10}\selectfont A \textsf{FlashFill} task to
  reformat the $2^{nd}$ date in an input.}
\label{figure:motivating_example}
\end{wrapfigure}
Consider the data transformation task shown in  
\figref{motivating_example}---the user wants to transform the input sentence 
fragments on the left to the output on the right.
The promise of the PBE paradigm is this:
if a user can provide a few examples of this transformation, the PBE synthesizer 
will automatically figure out the user's intent and produce a program that can
perform the transformation in general.


\begin{figure}
\vspace{-7ex}
\fontsize{8}{9}\selectfont
\begin{lstlisting}
def parse_datetime(x, regex_str):
  posix_format = { 
    "year": "%Y", "month": "%m", "day": "%d", "day_of_week": "%a"
    . . .
  }
  match = regex.fullmatch(regex_str, x)
  fmt_str, val_str = "", ""
  for k, v in match.groupdict().items():
    fmt_str += posix_format[k] + " "
    val_str += v + " "
  return datetime.datetime.strptime(val_str, fmt_str)

def transform(x):
  r1 = # regex for matching a comma and/or the string "and" 
       # surrounded by arbitrary amounts of whitespace.
  r2 = # regex for matching a date in multiple formats.
  date_start_index = # end index of first match of r1 in x
  date_end_index = # end index of last match of r2 in x
  date_string = x[date_start_index:date_end_index]
  input_date_format = r"(?<month>\d{1,2})/(?<day>\d{1,2})/(?<year>\d{4})"
  dt_obj = parse_datetime(date_string, input_date_format)
  return dt_obj.strftime("%B ") + "{0:01d}".format(dt_obj.day) +
         dt_obj.strftime(", %Y")
\end{lstlisting}
\vspace{-4ex}
\caption{Program $P_1$: \flashfill produced code for the task in \figref{motivating_example}.}
\label{figure:phaseone_program}
\vspace{-5ex}
\end{figure}
\begin{figure}
\fontsize{8}{9}\selectfont
\begin{lstlisting}
def transform(x):
  try:
    dt_obj = datetime.datetime.strptime("%m/%d/%Y", x[15:25]) 
    return dt_obj.strftime("%B ") + "{0:01d}".format(dt_obj.day) +
           dt_obj.strftime(", %Y")
  except ValueError:
    return None          
\end{lstlisting}
  \vspace{-4ex}
\caption{Program $P_3$ to perform task shown \figref{motivating_example}.}
\label{figure:phasethree_program}
  \vspace{-6ex}
\end{figure}

Providing a single example to the \flashfill PBE system produces the 
program in \figref{phaseone_program}, which correctly performs the user-intended
transformation.
This program is very unlike any program that would be written by a human 
programmer for the same task---a human programmer might write a program
that is closer to  the one in \figref{phasethree_program}.

\boldpara{Hurdles to adoption of PBE.}
The difference between these programs illustrate some of the major hurdles facing
a more wide-spread adoption of PBE among data scientists and programmers.
Not only is $P_3$ more readable and compact than $P_1$, but also significantly more
efficient.
Readable and efficient programs are more likely to be used in practice because:
\begin{compactenum}[(a)]
  \item when using PBE to perform data preparation or data processing, 
    a user is often directly paying for the computation time, and given
    the prevalence of multi-terabyte datasets and streaming data in this domain,
    even small improvements in performance correspond to big reduction in operating costs,
  \item compactness and readability of a program is a proxy for its maintainability,
      which is extremely desirable for expert programmers, who often do not trust 
      (PBE) systems that produce unreadable code.
\end{compactenum}

\boldpara{Why are PBE produced programs different?}
Before we address the question of making PBE produced programs more compact and 
efficient, let us first examine why PBE programs are the way they are.
We frame these reasons in the context of the differences between $P_1$ and $P_3$.

\itpara{PBE programs are general.}
The first significant difference between $P_1$ and $P_3$ is manner in which
the $2^{nd}$ date is extracted.
Program $P_1$ locates this substring by using a combination of searches for the
constant \texttt{"and"} surrounded by arbitrary white space and a
regular expression representing dates in different formats.
On the other hand, the $P_3$ just picks the sub-string between indices $15$ and $25$.

The difference here is generality: $P_1$ works on a larger variety of inputs.
For example, $P_1$ can handle the input \texttt{"CAV 2019 is between 23/07/2019 and
26/7/2019 and is in NYU."} in a correct manner, while $P_3$ clearly does not.
{\bf The PBE synthesizer has generalized the program to handle a large
variety of inputs, as it does not know the kind of inputs the program is
meant to handle} ahead of time.
On the other hand, the human programmer knows the input format, and has optimized
the program for it.

\itpara{PBE systems use a domain-specific language.}
The second significant difference between $P_1$ and $P_3$ is the handling of date-time
operations.
Program $P_3$ uses simple calls to the native Python date-time library, while $P_1$
uses a complex wrapper.
This wrapper is present because {\bf PBE synthesizers generate programs in a domain-specific 
language optimized for synthesis}, which are then translated to Python.
The DSL has its own operators because it needs to support efficient synthesis on the task
domain, which is best done by being agnostic to the target languages.
\ignore{ 
The use of a DSL enables two important aspects of PBE systems:
\begin{compactitem}
  \item[Efficient synthesis.] 
      Each synthesis technique imposes requirements on the target language for efficient
      synthesis. For example,
      \begin{inparaenum}[(a)]
        \item \flashfill and \flashmeta require that the \emph{inverse semantics} of operators be computable,
        \item SyGuS solvers require that the operators be encodable in an SMT decidable theory
      \end{inparaenum}, etc.
      No high-level language comes close to satisfying any of these requirements.
  \item[Multi-targeting.] A DSL lets the synthesizer avoid reasoning about the 
      idiosyncrasies and limitations of any language library.
      Instead, programs are produced in a DSL and then translated to different target languages
      (Python, Java, etc).
      On the other hand, this introduces the complications of the DSL operation semantics into
      the synthesized program.
\end{compactitem}
\endignore}

\itpara{PBE synthesizers optimize for user interaction.}
A third, more subtle reason why $P_1$ is different from $P_3$ is not apparent from
the programs themselves, but the process by which \flashfill produced $P_1$.
\flashfill picked $P_1$ over other similarly general programs due to its
\emph{ranking function}---PBE synthesizers often produce a large number of
candidate programs and pick one.
Most {\bf PBE ranking functions are optimized to converge in on user intent with 
the fewest examples}.
This UX oriented factor underlies both the generality of programs, the design of the DSL, as well
as other artifacts of PBE-produced programs.

\boldpara{Optimizing PBE-produced programs.}
We address the problem of optimizing PBE programs (in particular, for 
compactness and performance).
The three points mentioned above, in a manner of speaking, act as constraints 
to any PBE optimization procedure:
\begin{inparaenum}[(a)]
\item We want to minimize the number of examples required to converge to 
  the user intent. 
\item The ``PBE'' part of the optimization procedure has to operate on a DSL, 
  while optimal programs go beyond the DSL.
\item The ``generality towards inputs'' issue can only be solved by specifying
  the set of intended inputs.
\end{inparaenum}

To this end we use a three-phase optimization procedure:
the procedure is given a set of examples $E$ as usual, and in addition,
a set of inputs $\mathsf{pre}$ (explicit or symbolic) on which the synthesized 
program is expected to work on.
\begin{compactitem}
  \item The first phase is a standard PBE run to produce a program.
    This step solves the issue of minimizing user interaction.
    On our example task, \flashfill produces $P_1$.
  \item The program $P_1$ acts as an \emph{equivalence specification} on $\mathsf{pre}$ 
    for the second phase, i.e., we want the program generated by the second phase
    to be behaviorally equivalent to $P_1$ on inputs in $\mathsf{pre}$.
    In the second phase, we use a PBE synthesizer with a compactness or 
    performance-based ranking function to produce a program $P_{12}$.
    This $P_{12}$ is still a program in the DSL---it does not use optimal
    target language specific operators.
  \item In the third phase, $P_{12}$ is rewritten via local enumerative search;
    Sub-expressions in $P_{12}$ are rewritten with more optimal sub-expressions
    from the target language to obtain $P_{123}$.
    Again, $P_{123}$ needs to behave as $P_{12}$ on $\mathsf{pre}$.
\end{compactitem}
On our running example, $P_{12}$ (shown in Appendix~\ref{sec:app_2}) uses the 
more optimal method (indexing between $15$ and $25$) to select the date, but 
still uses DSL specific date-time operations.
These DSL specific operations are then rewritten to native Python function calls in
the third phase---in our evaluation, $P_{123}$ was identical to $P_{3}$ shown in 
\figref{phasethree_program}.
As per our measurements, $P_{123}$ is \textbf{2.79X} faster than $P_1$ and \textbf{1.62X}
faster than $P_{12}$ on the given data-set.

\section{The General Quantitative PBE Problem}
\label{section:general_problem}


The goal of {\em{Programming-by-example (PBE)}} is to synthesize a program
$p$ --- that transforms values from an input domain $D_i$ to values in an
output domain $D_o$ --- from an incomplete specification given as a
{\em{small}} set of input-output examples $\esmall$.
The synthesized program $p$ is expected
to work correctly not only on the inputs in $\esmall$, but also on a larger
set $\fullinputset$ of inputs (that includes the inputs in $\esmall$).  The
set $\fullinputset$ may be represented as an explicit enumeration of its
elements (as is the case in our experiments), or it could be represented
symbolically.  Let $\specification{\fullinputset}
: \fullinputset \rightarrow D_o$ be an unknown (or black-box) function
that can be queried to provide the correct (or user-intended) output for a
specific input $i \in \fullinputset$. Assume that the results of any
queries (on $\specification{\fullinputset}$) made by a PBE algorithm are
accumulated in the set of input-output
examples $\esmall \subseteq \fullinputset \times D_o$. In other words, we have for
every $(i, o) \in \esmall$, $\specification{\fullinputset}(i) = o$.
The objective of PBE is to synthesize a program that meets the
specification $\specification{\fullinputset}$, while minimizing
$|\esmall|$. 
It is important to minimize $|\esmall|$ because the user is expected to
play the role of $\specification{\fullinputset}$, and the goal is to
minimize the cognitive load on the user.


There is plenty of existing work that addresses the PBE challenge.  In this
paper, we are interested in an {\em{extension of the problem}} where the
user is not interested in just {\em{any}} program $p$ that is
correct on the input set $\fullinputset$, but a program that is
also \emph{optimal} with respect to some user-defined metric.
Specifically, the user wants a program $p$ in a {\em{real-world}} target
language (e.g., Python, Java, etc.), say $\lang$, that minimizes a given cost function $c$. The cost
function maps a program in the target language $\lang$ to a nonnegative
real number.  Thus, the user wants to minimize (a) the number of examples
in $\esmall$ that need to be provided, as well as (b) the cost $c(p)$ of the
synthesized program $p$.

The solution strategy that we propose in this paper is applicable to any
suitably defined quantitative cost metric. However, in this paper we
consider two specific cost metrics: {\em{performance}} and {\em{size}} of
the generated program in the target language Python.  The motivation for
performance is clear: users want to run the synthesized programs --- often
in cloud computing environments --- on large datasets --- which can contain
millions of rows --- and wish to minimize resource utilization, and hence
cost.  The motivation for minimizing size comes from generating programs
that are easier for the user to quickly read, understand, and possibly even
edit.  While size is not a sole contributor to a program's readability, it
is a well-defined quantitative metric that we use as a proxy for
readability, and as a first approximation.

Formally, we study the following general quantitative PBE (\qpbe) problem.
\begin{definition}[\qpbe]
\label{def:qpbe}
Let $\lang$ be a fixed target language. Let $c: \lang \rightarrow \RR^+$ be a
fixed cost function that maps programs in $\lang$ to a non-negative
cost. Let $D_i$ and $D_o$ be the domains of input and output values
respectively. Let $\fullinputset \subseteq D_i$ be a symbolic or explicitly
enumerated restriction (or precondition) on $D_i$. Let
$\specification{\fullinputset} : \fullinputset \rightarrow D_o$ be an
implicit and unknown function that describes the input-output behavior
of a desired program. Let $\esmall \subseteq \{(i, o) \mid i \in D_i, o
= \specification{\fullinputset}(i)\}$ be a small set of input-output
pairs that are obtained by invoking $\specification{\fullinputset}$ for
specific values $i \in \fullinputset$. Then, the {\em{\qpbe problem}} is to
find a program $p \in \lang$ such that
\\
(1) (Correctness) $p$ satisfies the specification $\specification{\fullinputset}$;
i.e, $\forall i \in \fullinputset . p(i) = \specification{\fullinputset}(i)$,
\\
(2) (PBE objective)  $|\esmall|$ is minimal, \ie, the number of queries
made to $\specification{\fullinputset}$ during the process of finding $p$
is minimal, and
\\
(3) (Cost objective) the cost $c(p)$ of the synthesized program $p$ is
minimal (among all programs that satisfy the correctness objective).
\end{definition}

Note that the \qpbe problem involves optimizing for two objectives, where
one objective is inherited from PBE, and the other is a cost objective. The
PBE objective is a requirement for the \emph{synthesis algorithm to be
effective}, whereas the cost objective is a requirement on the \emph{output
of the synthesis algorithm}. So, even though the \qpbe problem appears to
be a multi-objective optimization problem, the two objectives live in
different dimensions and we exploit this separation in our solution.

\begin{wrapfigure}[9]{R}{0.5\textwidth}
\vspace*{-11mm}
  \centering
  \scalebox{0.8}{
\begin{tikzpicture}
  \path
  node[input box, inner sep=2mm] (specification) {Examples\\$\esmall$}
  node[basic box=blue, right=4mm of specification, inner sep = 2mm] (synthesizer) {PBE\\Synthesizer\\\textsc{synthesize}}
  node[input box, above=4mm of synthesizer, inner sep=2mm] (dsl) {DSL $\dsllang$}
  node[input box, below=4mm of synthesizer, inner sep = 2mm] (ranker) {Ranking function $\rankingscheme$}
  node[basic box=white, right=4mm of synthesizer, draw=none] (program) {Top-ranked\\Program\\(w.r.t. $\rankingscheme$)}
  ;
  \path[data line,->]
  (ranker) edge[data line, solid] (synthesizer)
  (specification) edge[data line, solid] (synthesizer)
  (dsl) edge [data line, solid] (synthesizer)
  (synthesizer) edge [data line, solid] (program)
  ;
\end{tikzpicture}
  }
  \caption{\fontsize{8}{10}\selectfont An abstract representation of a PBE system}
  \label{figure:abstract_pbe_system}
\end{wrapfigure}
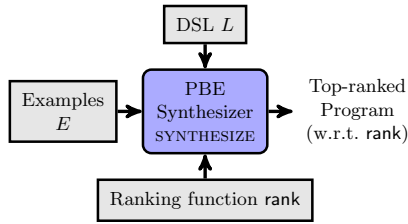
Our high-level approach for \qpbe relies on solvers for the PBE problem,
and hence we briefly\footnote{We provide a very high-level and informal
description of PBE systems for the sake of completeness. The reader is
referred to earlier work~\cite{polozov-15} for more details.} discuss the
workings of a PBE engine. PBE engines achieve the goal of learning from as
few examples as possible by using a {\em{ranking function}} that orders
programs (in the set of all programs that are consistent with the few
examples provided) by their estimated likelihood of being the user-intended
program. Hence, modern PBE systems can be viewed as learners that do not
leverage knowledge of $\fullinputset$, but only take a set of examples and
a ranker and return the highest ranked program that is consistent with the
set of examples.

\subsection{An Overview of PBE Systems}

A PBE system (shown in \figref{abstract_pbe_system}) is parameterized by a
domain specific language (DSL) $\dsllang$, and a ranking function
$\rankingscheme$. Given a set of input-output examples $\esmall$,
the PBE synthesizer returns the top-ranked program (with respect to
$\rankingscheme$) in $\dsllang$ that satisfies all of the input-output
examples in $\esmall$.

\paragraph{Domain Specific Language (DSL) $\dsllang$.}
A DSL is used to restrict the search space for a synthesizer and
consists of a set of operators along with a grammar.
We assume that the DSL is specified as a
context free grammar (CFG) with a designated \emph{start symbol}. Given a
CFG $L$, we define $\prog(L)$ to be the set of all programs derivable from
the start symbol of $L$. Further, we define $\subprog(L)$ to be the set of
all (sub)programs derivable from \emph{any} terminal or non-terminal symbol
in $L$.

\begin{wrapfigure}[11]{R}{0.5\textwidth}
  \vspace*{-7mm}
  \fontsize{8}{10}\selectfont
  \setlength{\tabcolsep}{2pt}
  \centering
  \begin{tabular}{lcl}
    \gtype{string} \gsymbol{start} & \expandsto & \gsymbol{e};\\
    \gtype{string} \gsymbol{e} & \expandsto & \gsymbol{f} \!$\vert$\! \goperator{concat}\!\!(\gsymbol{f},\! \gsymbol{e}\!\!);\\
    \gtype{string} \gsymbol{f} & \expandsto  & \goperator{conststr}\!\!(\gsymbol{s})\\
    & $\vert$ & \goperator{substring}(\gsymbol{in}, \gsymbol{pp});\\
    \gtype{(int, int)} \gsymbol{pp} & \expandsto & (\gsymbol{pos}, \gsymbol{pos})\\
    \gtype{int} \gsymbol{pos} & \expandsto & \gsymbol{k} \!$\vert$\! \goperator{regexpos}\!\!(\gsymbol{in},\! \gsymbol{r},\! \gsymbol{r},\! \gsymbol{k}\!\!);\\
    \gtype{int} \gsymbol{k}, \ \gsymbol{idx};\\
    \gtype{regex} \gsymbol{r};\\
    \gtype{string} \gsymbol{s};\\
    \lbrack\gtype{input}\!\!\rbrack\gtype{string} \gsymbol{in};
  \end{tabular}
  \caption{\fontsize{8}{10}\selectfont A simplified version of the \textsf{FlashFill} DSL}
  \label{figure:flashfill_dsl}
\end{wrapfigure}
\itpara{Example: The \flashfill Grammar.}
\figref{flashfill_dsl} shows a simplified version of the \flashfill DSL for
synthesizing text transformations, given a single input \texttt{string} and
producing a single \texttt{string} as output. The start symbol of this grammar is \emph{e}, which can expand to
either a single string \emph{f} or a concatenation of two or more strings
(via the rule \emph{e} := \goperator{concat}(\emph{f}, \emph{e})). The
symbol \emph{f} can in turn expand to either a constant string (represented
as $\goperator{constr}(s)$), or a substring of some string in the
input. The symbol \emph{ss} represents a substring of a string from the
input array,
which can be computed using the start and end positions (represented by the
symbol \emph{pp}). The positions themselves can either be absolute positions or
regular expression based positions. Finally, the symbols \emph{k},
\emph{idx}, \emph{r} and \emph{s} are terminal symbols, and \emph{in} is
the designated input symbol.

\paragraph{Ranking Function.}
A \emph{ranking function} \rankingscheme for a DSL \dsllang is defined as a map 
from $\subprog(\dsllang)$ to $\RR^+$.
Ranking functions
are used in PBE systems to impose some ordering on the \dsllang-(sub)programs.
They may also be used to prune parts of search space of possible programs when only a
top-ranked program is desired. Recall that a PBE system is expected to work
with a partial specification.
When given only a few examples, there may exist a large number of
semantically distinct programs in $L$ that are all consistent with the given 
examples. The default ranker in a PBE system is often highly-tuned to choose 
\emph{one} program that is most likely to produce results a user might
expect when executed on new and unseen data; that is, the default rankers are
designed to work with minimal example sets $\esmall$, while still generating
programs that meet $\specification{}$.


\subsection{A High-Level Approach for \qpbe}\label{subsec:meta}

\ignore{
The \qpbe problem in Definition~\ref{def:qpbe} is still too abstract and general to admit
a concrete automated solution. We specialize the problem in this section, but we do so
using a very general high-level approach. We believe that the general approach can also
be used on classes of \qpbe problem that are different from those studied in this paper.
\endignore}

Our high-level approach consists of reusing existing PBE engines to handle the PBE objective,
as well as a part of the cost objective, and using rewriting to eventually optimize the cost
on the target language $\lang$.
Since existing PBE technology is unlikely to work directly on the grammar for $\lang$, we design a 
domain-specific language $\dsllang$ to perform synthesis, and then extend the solution on $\dsllang$ 
to a solution over $\lang$.
In detail, the high-level approach starts with the \qpbe problem from Definition~\ref{def:qpbe}
and performs the following steps {\em{manually}}:
\begin{compactenum}
    \item Design a DSL, $\dsllang^e$, whose programs translate to a
        program in $\lang$ using a translator, $\translate: \prog(\dsllang^e) \mapsto \lang$.
        The translator is semantics preserving.
    \item 
        The DSL $\dsllang^e$ may not be ideal for synthesis using a PBE engine. 
        Pick a sub-DSL, $\dsllang$,  of $\dsllang^e$, such that 
        synthesis can be efficiently performed on $\dsllang$ using a PBE
        engine. One can view $\dsllang^e$ as extending $\dsllang$ with
        constructs and standard library routines of the target language
        $\lang$ that will be useful in translating a program in $\prog(\dsllang)$
        to a program in $\lang$.
    \item Design a ranker $\rankpbe: \subprog(\dsllang)\mapsto\RR^+$ that is optimized 
        for solving PBE on $\dsllang$. In other words, $\rankpbe$ is
        designed to find the user-intended program (in $\prog(\dsllang)$)
        using as few examples as possible.
    \item Using the cost metric $c$, which is defined on programs in $\lang$, design a 
        ranker $\rankc: \subprog(\dsllang^e)\mapsto\RR^+$ such that, for
        all $p,p'\in\prog(\dsllang^e), \rankc(p) > \rankc(p') \Rightarrow
        c(\translate(p)) < c(\translate(p'))$.
    \item 
        Design a set $\RRules$ of rewrite rules from programs in
         $\dsllang$ to programs in $\dsllang^e$ such that with each
         rule $\rho \equiv p\rightarrow p'$, there is an
         associated precondition $\psi_\rho$ such that
         $p\models\specification{} \Rightarrow
         p'\models\specification{}$ whenever the inputs in
         $\specification{}$ satisfy the precondition $\psi_\rho$.
\end{compactenum}
Note that the rewrite rules are not semantic-preserving in general, but
they are designed so that whenever the inputs satisfy some precondition,
the rewritten program's behavior matches that of the original program.
In theory, we can include any rewrite rule in $\RRules$ since $\psi_\rho$ can
always be picked to be the empty set in the worst-case. 
In practice, it may be desirable to restrict the size of $\RRules$ and
the nature of the rewrite rules in $\RRules$ to ensure an efficient
implementation. 

These steps are performed manually for the fixed target language $\lang$
and cost metric $c$.
In many cases, one can re-use existing DSLs. 
For example, when the target language changes, we could reuse the same
$\dsllang^e$, $\dsllang$, and the rewrite rules $\RRules$, and just change the
translator $\translate$.  When solving a new \qpbe problem, in practice, we
perform the meta-procedure above as follows: (1) use an existing
DSL and its default ranker as $\dsllang$ and $\rankpbe$
respectively, (2) write a baseline translator $\translate$ to go from
$\dsllang$ to the target language $\lang$, (3) in the first attempt, use
$\dsllang$ also as $\dsllang^e$, and empty set $\RRules=\emptyset$ of rewrite
rules, (4) design the ranking function $\rankc$ based on cost
metric $c$.  We start with this solution, and experimentally evaluate the
difference between the program generated using the above choices
and the desired program.
The desired program can be obtained by mining source code repositories,
using user studies, or plain manual inspection.
If the desired output programs contain language
features that do not have a direct analogue in $\dsllang$, then we extend
$\dsllang$ to $\dsllang^e$ and write (unsound) rewrite rules for going from
$\dsllang$ to $\dsllang^e$, while also extending the translator and
$\rankc$. 
After the above meta-steps, we get a specific and concrete \qpbe problem.

\subsection{\qpbe Modulo PBE}
In this section, we define a concrete and specific class of \qpbe problem.
We assume that we are given a PBE engine, $\mathsf{PBE}$, that is
parameterized by a DSL and a ranking function (as shown
in \figref{abstract_pbe_system}).
\begin{definition}[\qpbe/PBE]\label{def:qpbe-pbe}
    The \qpbe/PBE problem instance is an instance of the \qpbe problem, defined in Definition~\ref{def:qpbe},
    where we are additionally given the following artifacts designed based on the target language $\lang$ and
    cost metric $c$:
    (1) the DSL $\dsllang$ and its extension $\dsllang^e$,
    (2) a ranker $\rankpbe$ on $\dsllang$ that optimizes for 
    the number of examples required to disambiguate user-intent,
    (3) a ranker $\rankc$ on $\dsllang^e$ that optimizes for
    the cost of the translated program,
    (4) a set $\RRules$ of rewrite rules from programs in $\dsllang$ to programs in $\dsllang^e$, and
    (5) a translator $\translate$ from $\dsllang^e$ to the target language $\lang$.
\end{definition}

\vspace{-2ex}
\section{An Algorithm for \qpbe/PBE}
\label{section:general_framework}
\vspace{-1ex}
We describe our solution to the \qpbe/PBE problem.  In the rest of this
paper, we assume that a PBE system can be used to synthesize the top-ranked
program in a DSL $\dsllang$ given input-output examples $\esmall$
and a ranking function $\rankingscheme$. Specifically, we assume we have
access to a procedure call
$\mbox{\textsc{synthesize}}(\esmall, \dsllang, \rankingscheme)$,
which returns the top-ranked program in $\prog(\dsllang)$ with respect
to $\rankingscheme$ that satisfies $\esmall$.

\subsection{Three-Phase Algorithm for \qpbe/PBE}

\begin{algorithm}[!t]
  \caption{\qsynth: Synthesize a optimal program by examples.}
  \label{algorithm:general_algorithm}
  \begin{small}
    \DontPrintSemicolon
    \SetKwInOut{Inputs}{Inputs}
    \SetKwInOut{Output}{Output}
    \SetKwInOut{Data}{Data}
    \SetKwRepeat{Do}{do}{while}
    \Inputs{
        $\dsllang$, $\dsllang^e$, $\translate$, $\rankpbe$, $\rankc$,
        $\RRules$, $\esmall$, $\fullinputset$
    }
    \Output{
      A program $P$ as described in Definition~\ref{def:qpbe}.
    }
    \vspace*{2mm}
    $p_1 \leftarrow {\textsc{synthesize}}(\esmall, \dsllang, \rankpbe)$\;
    \label{line:p1}
    \If{$p_1 \equiv \bot$}{
      \label{line:early_fail}
      \Return $\bot$\;
    }
    ${\esmall}_{\equiv} \leftarrow \left\{(i, p_1(i)) : i \in \mathsf{representative\_sample}(\fullinputset)\right\}$\;
    \label{line:build_equiv_spec}
    $p_2 \leftarrow {\textsc{synthesize}}(\esmall_\equiv, \dsllang, \rankc)$\;
    \label{line:p2}
    $p_3 \leftarrow {\textsc{enumerative\_synth}}(\esmall_\equiv, \dsllang^e, \rankc, p_2, \RRules)$\;
    \label{line:p3}
    \Return $\translate(p_3)$\;
    \label{line:end}
  \end{small}
\end{algorithm}

The procedure for \qpbe/PBE, presented in
\algoref{general_algorithm}, consists of three phases.
The objective of the first phase, {\em{Intent Disambiguation}}, is to
obtain a close approximation (in the form of a large number of input-output
examples) of the complete specification $\specification{\fullinputset}$
using as few examples as possible.  The second phase,
{\em{Global Search}}, finds a candidate program that is ``close'' to
being optimal from the global search space.  The third phase, {\em{Local
Search}}, performs local rewrites to the program found in the second phase
to yield an optimal program.

\paragraph{Intent Disambiguation.}
In the first phase, we use the PBE engine to solve the PBE problem
contained in the \qpbe problem.  Specifically, we synthesize a program
$p_1$ over the DSL $\dsllang$ while minimizing the number of input-output
examples in $\esmall$, using a ranker $\rankpbe$ optimized for learning the
user-intended program (i.e. program that satisfies the user-intended
specification $\specification{\fullinputset}$).  The
PBE engine can fail to find a program, in which case the whole process
terminates with failure.  This happens when there is no program in
$\prog(\dsllang)$ that satisfies all the input-output examples in
$\esmall$.  However, when this is not the case, the PBE engine will find a
program $p_1$. Now, $p_1$ is the program predicted to be \emph{most likely
to match the user's intent} by the ranker $\rankpbe$. Thus, any program
$p'$ that is \emph{behaviorally equivalent} to $p_1$ on $\fullinputset$
should also match the user's intent just as well. We leverage this
intuition, and use $p_1$ to create a set of input-output examples that
comprises an \emph{equivalence specification} $\esmall_{\equiv}$: the
output corresponding to each input $i \in \fullinputset$ is simply
$p_1(i)$.

However, executing $p_1$ on the \emph{entire} set $\fullinputset$ defeats
the point: why bother with finding a better program if the desired data
transformation task has already been accomplished, by running $p_1$ on the
entire set of inputs $\fullinputset$? To avoid this, we instead construct
$\esmall_{\equiv}$ by executing $p_1$ on a \emph{small and representative
sample} of $\fullinputset$, as shown in line~\ref{line:build_equiv_spec} in
Algorithm~\ref{algorithm:general_algorithm}. A representative sample can
be obtained by first clustering the input data using a technique
like \flashprofile~\cite{matching-text}, and then performing a stratified sampling over the
clusters.


\paragraph{Global Search.}
We could translate $p_1$ to a program $\translate(p_1)$ in $\lang$.
However, the cost $c(p_1)$ of this program is unlikely to be low.
After all, $p_1$ was derived completely
independently of $c$.
Hence, in the second phase, we use the
ranking function $\rankc$ (which is based on the cost function $c$) to do a
global search over $\prog(\dsllang)$ to find a program $p_2$ which is
behaviorally equivalent to $p_1$, but is optimal with respect to $c$, as
shown in line~\ref{line:p2} in Algorithm~\ref{algorithm:general_algorithm}.
Note that the PBE engine can use any method to perform synthesis;
for instance, even though it has access to a large set of input-output
examples, the PBE engine can still use the Counter-example Guided Inductive
Synthesis~\cite{solar-lezama-05} paradigm and incrementally expand the set
of examples actually used in the synthesis.

\paragraph{Local Search.}
The last phase of our quantitative PBE synthesis procedure involves
rewriting the program $p_2$ computed by the global search phase to an
optimal program in the DSL $\dsllang^e$.  Recall that $p_2$ is a program in
DSL $\dsllang$, but the target language is not $\dsllang$. The grammar
$\dsllang^e$, which is closer to the grammar for the target language
$\lang$, may contain function symbols that have no direct counterparts in
$\dsllang$. Hence, going from $\dsllang$ to $\dsllang^e$, exposes
opportunities for optimization, which are leveraged in the third phase.  We
call this phase local search since it does not significantly change the
overall logic of the program $p_2$, but only maps it into a program
$p_3 \in \prog(\dsllang^e)$ such that $p_3$ can be translated into a higher
ranked program in $\lang$, than $p_2$, with respect to $\rankc$.

A final key and interesting insight is that the third phase is yet another
synthesis problem, but with a slightly different formulation. Recall that
we have a set of rewrite rules $\RRules$ that consist of possible
transformations that may be done on a program when going from $\dsllang$ to
$\dsllang^e$.

The search space of possible programs is defined as the set of all programs
reachable from the {\em{starting term}} $p_2$ using the rewrite rules
$\RRules$. From this space, we need to find one that
satisfies $\esmall_\equiv$ and is highest ranked by the ranker $\rankc$.

\begin{definition}[Local \qpbe]\label{def:local-qpbe}
    Given a signature $F$ of function symbols and constants, an initial
    term $t$ constructed using the signature $F$, rewrite rules $\RRules$
    that transform a term to another term, a ranker $\rankc$ that maps
    terms to rationals, and an input-output example based equivalence
    specification $\esmall_{\equiv}$, Local \qpbe seeks to find a term $t'$
    reachable from $t$ using zero or more applications of rules in $RR$
    such that (1) $t'$ satisfies the equivalence specification
    $\esmall_{\equiv}$, and (2) the term $t'$ is highest ranked term in the
    reachable set.
\end{definition}

The symbols in $F$ are assumed to have executable (operational) semantics.
This allows us to determine if a given term meets a given specification by
just computing the interpretation of the term on the input.

A naive procedure for solving the local \qpbe problem is as follows:
enumerate the reachable terms and prune out the terms that are either not correct
(do not satisfy the specification) or are lower in rank.
Some PBE engines that support enumerative search (bottom up synthesis) can be adapted
to solve the local \qpbe synthesis problem.
Let ${\textsc{enumerative\_synth}}(\esmall_{\equiv}, \dsllang^e, \rankc, p, \RRules)$
denote a enumerative search based synthesis procedure that solves the local \qpbe problem.
We use this procedure to perform the final local search. If $p_3$ is the program in $\dsllang'$
synthesized by the enumerative search procedure, then we return 
$\translate(p_3)$ as the final answer to the user.

\ignore{
\vspace*{1mm}
The procedure detailed above overcomes all three design challenges
for a PBE optimization procedure mentioned in Section~\ref{section:motivating_example}/

\emph{The intent disambiguation} phase lets the user provide very few
examples, and perform synthesis
using $\rankc$ from the get-go instead?
This results in a user having to provide a larger number of examples.

\vspace*{1mm}
\boldpara{What ifs, buts, and why nots.}
All the answers in this section are validated by experimental results
presented in \secref{experimental_evaluation}, see RQ2 in particular.
}

\begin{remark}
\emph{Why not simply perform the global search over $\dsllang^e$, and skip the local search phase?}
In general, synthesis over $\dsllang^e$ may not be feasible.
First, a general purpose $\dsllang^e$ will induce a significantly larger
search space.
Second, most program synthesis techniques require additional properties 
over the language: For example, \flashfill and \flashmeta
require operators to have inverse semantics~\cite{polozov-15}; Sketch and
SyGuS solvers require operators to be encodable in a decidable SMT theory.

\emph{What if $\mathsf{representative\_sample}(\fullinputset)$ does not
yield a truly representative sample?}
In this case, our approach (and any PBE approach for that matter) may not
yield a program that matches the user's intent.
PBE systems provide very weak guarantees in general:
the end-user is the best judge of correctness.
\end{remark}

\vspace{-4ex}
\subsection{Correctness}\label{sec:correctness}
We denote by $\specification{\sample}$, the specification
$\specification{\fullinputset}$, restricted to a representative sample
of $\fullinputset$, \ie, $\specification{\sample}(i)
= \specification{\fullinputset}(i)$ iff
$i \in \mathsf{representative\_sample(\fullinputset)}$ and $\bot$ otherwise.
We first note that the quantitative synthesis procedure is sound: the output
program satisfies the specification $\specification{\sample}$.
\begin{proposition}\label{prop:soundness}
    If the PBE synthesis procedure, \textsc{synthesize}, and the
    enumerative search procedure, \textsc{enumerative\_synth}, both return
    a program that satisfies the specification, then the output of the
    quantitative synthesis procedure also satisfies the specification.
    Furthermore, the number of input-output examples used by the
    quantitative synthesis procedure shown in \algoref{general_algorithm} is
    equal to the number used in the first phase by the underlying PBE
    engine.
\end{proposition}

\ignore{ 
The proof follows from the soundness assumption on the individual
components.  Also note that the equivalence specification $\esmall_\equiv$
computed on Line~\ref{line:build_equiv_spec} is equivalent to the intended
program's specification.  Furthermore, the only input-output examples
required are those in intent disambiguation phase (on Line~\ref{line:p1})
since the latter phases already know all the input-outputs there are to
know.
\endignore}

In general, the final output program of
\textsc{synthesize} is not guaranteed to be the minimum cost correct  program.
However, under some reasonable assumptions, we can still make certain
completeness claims.  A ranker $\rankc$ is {\em{monotonic}} if $\rankc(p_1)
< \rankc(p_2)$ implies $\rankc(p[p_1]_{\mathtt{pos}})
< \rankc(p[p_2]_{\mathtt{pos}})$ for all $p,p_1,p_2\in\subprog(L)$ and
$\mathtt{pos}\in\mathtt{Pos}(p)$.  The notation $p[p_1]$ denotes the new
subprogram obtained by replacing the subprogram $p|_{\mathtt{pos}}$ of $p$
at position $\mathtt{pos}$ by $p_1$. The set $\mathtt{Pos}(p)$ denotes the
set of all positions in the term $p$ (defined in the usual way). We use
$p[p_1]$ to denote $p[p_1]_{\mathtt{pos}}$ when the position $\mathtt{pos}$
is clear from the context.
\begin{proposition}\label{prop:completeness}
Suppose all function symbols in the signature $F$ of DSL $\dsllang$
have complete witness functions~\cite{polozov-15}.  If the language
generated by $\dsllang$ is finite, and if the ranking function
$\rankingscheme$ is monotonic, then, for any set of input-output examples
$\esmall$,
\textsc{synthesize}$(\esmall, \dsllang, \rankingscheme)$ returns the
highest ranked (with respect to $\rankingscheme$) correct program, and
consequently, Program~$p_2$ computed on Line~\ref{line:p2}
of \algoref{general_algorithm} is the highest ranked --- with respect to
$\rankc$ --- correct program in
$\dsllang$.
\end{proposition}

\ignore{ 
\begin{proof}
    Suppose a PBE engine is trying to synthesize a program generated by symbol $q$
    that meets specification $\specification{}$.
    If $q \rightarrow f(q_1,\ldots,q_n)$ is a production rule of the DSL $\dsllang$,
    then a PBE engine can guess that the top function symbol in the desired program is
    $f$ and then use the witness function for $f$ to generate new synthesis problems
    (new specifications) on symbols $q_1, \ldots, q_n$.
    This divide-and-conquer approach for synthesis will terminate since the DSL $\dsllang$
    is assumed to contain no loops (it generates a finite language).
    Furthermore, since the signature $F$ and states $Q$ are finite, there are only 
    finitely many choices for the top symbol $f$.  Since the witness functions are assumed
    to be complete, no solution is ever lost when diving a problem to subproblems.
    Finally, a PBE engine can prune solutions based on their ranks while
    searching (see Polozov, et. al.~\cite{polozov-15} for details).
    However, since the ranker is monotonic, pruning will never eliminate the best ranked
    candidate: if the best program generated by $q$ has $f$ at the top, then its first
    child will be the best program generated by $q_1$ due to monotonicity, and hence,
    other programs can be safely pruned at $q_1$.
\end{proof}
\endignore}

Finally, we can derive the following result about the correctness
of \algoref{general_algorithm} from Propositions~\ref{prop:soundness}
and~\ref{prop:completeness}.
\begin{theorem}\label{theorem:main}
    Under the assumptions of Proposition~\ref{prop:soundness} and Proposition~\ref{prop:completeness},
    and the assumptions about the inputs outlined in the meta-procedure in Section~\ref{subsec:meta},
    \algoref{general_algorithm} is terminating, and 
    the output $prog = \translate(p_2)$ of \algoref{general_algorithm} is a program in 
    language $\lang$ that satisfies the
    specification $\specification{}$, and has a cost that is minimum in the set
    \begin{eqnarray*}
     \{ c(p) \mid p = \translate(t), t \models \specification{}, t\in \mathsf{SearchedSpace}\} & \mathrm{where} &
        \\
        \mathsf{SearchedSpace} = \{ t \mid t\in \prog(\dsllang) \} \cup \{ t \mid t\in \prog(\dsllang^e), p_2 \rightarrow_{\RRules}^* t\} & &
    \end{eqnarray*}
\end{theorem}

We finally note an efficient way to perform \textsc{enumerative\_search} in the case when the rewrite
rules $RR$ satisfy certain conditions.
\begin{proposition}\label{prop:convergent}
    If every $p\in \lang$ has a preimage $p^-\in\prog(\dsllang^e)$ s.t. $\translate(p^-) = p$,
    and if the range of $\rankingscheme_2$ contains only fixed precision rationals, then
    the procedure \textsc{enumerative\_search} is terminating. Furthermore, 
    if the subset of rewrite rules $\RRules$ that are rank increasing (cost decreasing) is confluent, 
    then the procedure can be performed efficiently in time $O(k |\RRules|)$, where $k$ is the number
    of successful rewrites applied by the procedure and $|\RRules|$ is the size of rewrite system.
\end{proposition}

\ignore{ 
\begin{proof}
    The rule applications are terminating since they increase the rank by at least a fixed constant
    (given by the precision), and ranks are upper-bounded by the rank of the preimage of the 
    lowest cost program.  Since rank-increasing rewrite rules are terminating, if they are also
    confluent, then the rules are convergent and define unique normal forms that are always reachable.
    Hence, the order of application of the rules does not matter.
\end{proof}
\endignore}

\section{Experimental Evaluation}
\label{section:experimental_evaluation}
\newcommand{\numBenchmarks}{701}
\newcommand{\numPerfPhaseOneTwoThreeIncrease}{610}
\newcommand{\numPerfPhaseOneTwoThreeDecrease}{91}
\newcommand{\numPerfPhaseOneTwoThreeOverPhaseOneTwo}{537}
\newcommand{\numPerfPhaseOneTwoThreeOverPhaseOneThree}{481}

\newcommand{\numReadPhaseOneTwoThreeIncrease}{493}
\newcommand{\numReadPhaseOneTwoThreeEqual}{195}
\newcommand{\numReadPhaseOneTwoThreeDecrease}{13}
\newcommand{\numReadPhaseOneTwoThreeOverPhaseOneTwo}{135}
\newcommand{\numReadPhaseOneTwoThreeOverPhaseOneThree}{395}

We evaluated our proposed methodology with 
respect to two quantitative objectives --- program performance and program
size.
The base PBE synthesizer for our evaluation is the \flashfill PBE system. 
Our benchmarks consist of $\numBenchmarks$ text transformation tasks derived
from both academic and industry sources.
The target programs for these tasks include string, sub-string, regex operations,
as well as operations on number and date parsing and formatting.

Our experiments are designed to answer the following research questions.
\begin{compactitem}
    \item[\textbf{RQ1}:] 
        How effective is Algorithm~\ref{algorithm:general_algorithm} at optimizing the two
        objectives?
        Informally, what is the improvement seen in program size and program
        performance on using Algorithm~\ref{algorithm:general_algorithm} over standard PBE?
    \item[\textbf{RQ2}:] 
        What is the importance of each phase of Algorithm~\ref{algorithm:general_algorithm}?
        Through ablation studies, we examine whether any of the phases of
        Algorithm~\ref{algorithm:general_algorithm} can be skipped, while obtaining the same
        performance and size gains.
    \item[\textbf{RQ3}:]
        What is the overhead of using Algorithm~\ref{algorithm:general_algorithm}
        over standard PBE?
\end{compactitem}

\paragraph{Methodology and Measurements.}
Our experimental set-up is presented in Figure~\ref{fig:methodology}.  For
each benchmark, the intent of the user (in the form of examples) is fed
into \flashfill (the \emph{intent disambiguation} phase, which we refer to
as Phase $1$) to produce the program $p_1$.  As per the algorithm, $p_1$ is
processed through the \emph{global search} and the \emph{local search}
phases (henceforth referred to as Phases $2$ and $3$, respectively), in
sequence, to obtain $p_{12}$ and $p_{123}$.  The Phase $2$ ranking models
for both objectives were hand written, with less than $1$ person-day of
development put into each of them.  For the performance and
readability objectives, we generate exactly $1$ and $5$ programs
respectively in Phase $2$.

For $p_{1}$ and $p_{123}$, the value of the objective $\mathbf{o_1}$ (baseline)
and $\mathbf{o_{123}}$ (QPBE) was measured.
Further, the time to run Phase $1$ was recorded as the \emph{standard PBE
time} $\mathbf{t_{pbe}}$, and the time to run Phase $2$ and Phase $3$ together 
was recorded as \emph{optimization time} $\mathbf{t_{opt}}$ --- the total running
time is $t_{pbe} + t_{opt}$.

In addition, we perform ablation studies by skipping Phase $1$, Phase $2$, and 
Phase $3$.
To measure the effect of skipping Phase $1$, we measure the number of examples $\mathbf{e_1}$
and $\mathbf{e_2}$ required to converge to the user intent by Phase $1$ and Phase $2$,
respectively.
When skipping Phase $2$ and Phase $3$, we measure the objective function on the
program $p_{13}$ (Phase $3$ run directly on $p_1$) and $p_{12}$ as 
$\mathbf{o_{13}}$ and $\mathbf{o_{12}}$, respectively.

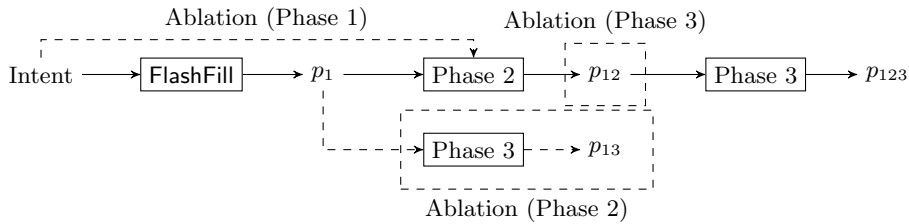
\begin{figure}
    \vspace{-4ex}
    \begin{tikzpicture}[scale=0.8]
        \node  (n0) {Intent};
        \node[draw, rectangle, right of=n0, xshift=1cm] (n1) {\flashfill};
        \node[rectangle, right of=n1, xshift=0.75cm] (p1) {$p_1$};
        \node[draw, rectangle, right of=p1, xshift=1cm] (n2) {Phase $2$};
        \node[rectangle, right of=n2, xshift=0.75cm] (p12) {$p_{12}$};
        \node[draw, rectangle, right of=p12, xshift=1cm] (n3) {Phase $3$};
        \node[rectangle, right of=n3, xshift=0.75cm] (p123) {$p_{123}$};

        \node[draw, rectangle, below of=n2] (n4) {Phase $3$};
        \node[rectangle, right of=n4, xshift=0.75cm] (p13) {$p_{13}$};

        \node[fit=(n4)(p13), inner sep=0.3cm, draw, dashed, label=below:{Ablation (Phase $2$)}] (f1) {} ;
        \node[fit=(p12), inner sep=0.2cm, draw, dashed, label=above:{Ablation (Phase $3$)}] (f2) {} ;

        \path[draw, ->] (n0) -- (n1) ;
        \path[draw, ->] (n1) -- (p1) ;
        \path[draw, ->] (p1) -- (n2) ;
        \path[draw, ->] (n2) -- (p12) ;
        \path[draw, ->] (p12) -- (n3) ;
        \path[draw, ->] (n3) -- (p123) ;

        \path[draw, dashed, ->] (p1.south) |- (n4.west) ;
        \path[draw, dashed, ->] (n4) -- (p13) ;

        \path[draw, dashed, ->] (n0.north) -- +(0, 0.3cm) node[above, xshift=3cm] {Ablation (Phase 1)} -|  (n2.north) ;
    \end{tikzpicture}
    \vspace{-5ex}
    \caption{Experimental methodology}
    \label{fig:methodology}
    \vspace{-8ex}
\end{figure}

\begin{figure}
    \vspace{-2ex}
    \centering
    \hspace*{-14mm}\begin{subfigure}[b]{0.43\textwidth}
        \includegraphics[scale=0.85]{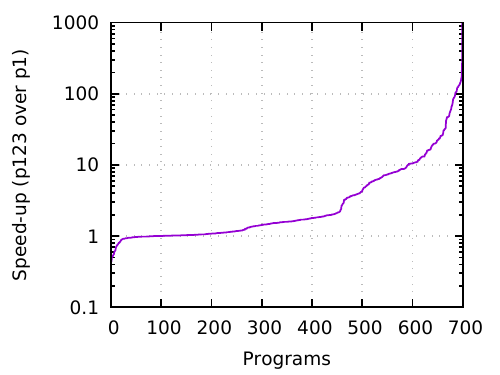}
        \caption{Algo.~\ref{algorithm:general_algorithm} effectiveness}
        \label{fig:perf:o123_vs_o1}
    \end{subfigure}
    \hspace*{-12mm}\begin{subfigure}[b]{0.43\textwidth}
        \includegraphics[scale=0.85]{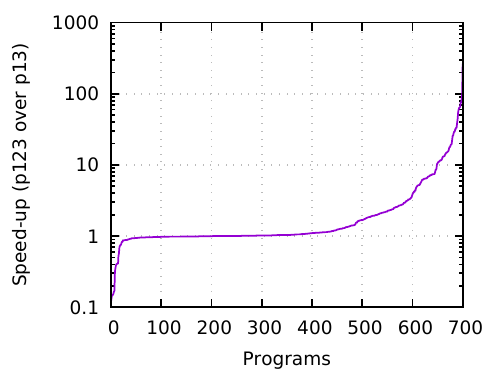}
        \caption{Phase $2$ ablation}
        \label{fig:perf:o123_vs_o13}
    \end{subfigure}
    \hspace*{-12mm}\begin{subfigure}[b]{0.3\textwidth}
        \includegraphics[scale=0.85]{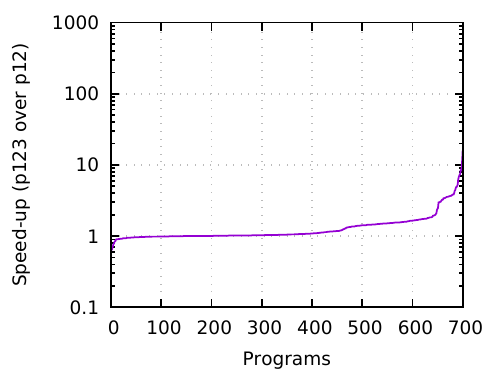}
        \caption{Phase $3$ ablation}
        \label{fig:perf:o123_vs_o12}
    \end{subfigure}
    \vspace{-2ex}
    \caption{
        Effectiveness of Algorithm~\ref{algorithm:general_algorithm} and individual phases for performance
     }
     \label{fig:perf}
    \vspace{-4ex}
\end{figure}

\begin{figure}
    \vspace{-4ex}
    \centering
    \hspace*{-14mm}\begin{subfigure}[b]{0.43\textwidth}
        \includegraphics[scale=0.85]{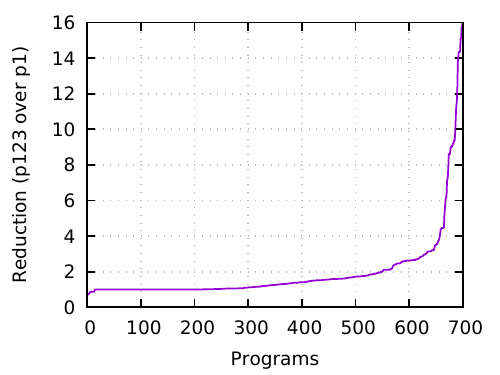}
        \caption{Algo.~\ref{algorithm:general_algorithm} effectiveness}
        \label{fig:read:o123_vs_o1}
    \end{subfigure}
    \hspace*{-12mm}\begin{subfigure}[b]{0.43\textwidth}
        \includegraphics[scale=0.85]{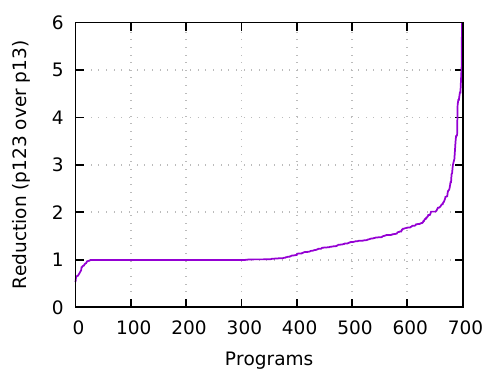}
        \caption{Phase $2$ ablation}
        \label{fig:read:o123_vs_o13}
    \end{subfigure}
    \hspace*{-12mm}\begin{subfigure}[b]{0.3\textwidth}
        \includegraphics[scale=0.85]{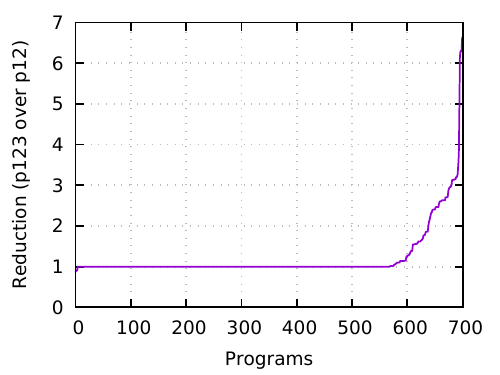}
        \caption{Phase $3$ ablation}
        \label{fig:read:o123_vs_o12}
    \end{subfigure}
        \vspace{-2ex}
    \caption{
        Effectiveness of Algorithm~\ref{algorithm:general_algorithm} for program size
     }
     \label{fig:read}
    \vspace{-4ex}
\end{figure}

\paragraph{RQ1: Effectiveness of Algorithm~\ref{algorithm:general_algorithm}}
Qualitatively, the value of program size and performance objectives
improve over standard PBE (i.e., $o_{123} > o_1$) for $\numReadPhaseOneTwoThreeIncrease$ and
$\numPerfPhaseOneTwoThreeIncrease$ of $\numBenchmarks$ benchmarks, respectively.
The overall profile of the improvement is plotted in Figure~\ref{fig:perf:o123_vs_o1} and 
Figure~\ref{fig:read:o123_vs_o1}, respectively.

\boldpara{Performance.}
Of the $101$ benchmarks where Algorithm~\ref{algorithm:general_algorithm}
did not improve performance, only $19$ benchmarks had throughput decreased by
$15\%$ or more, and only $2$ by $50\%$ or more.
Most cases where performance decreases are due to the Phase $2$ performance model 
being inaccurate with respect to real world performance.
%
Of the $\numPerfPhaseOneTwoThreeIncrease$ benchmarks where 
Algorithm~\ref{algorithm:general_algorithm} improves performance,
the most spectacular improvement was a benchmark with $3400X$ improvement due to the
replacement of multiple regex matching operation on long string inputs with a constant
index and a string find operation.
The remaining speed-ups were less than $300X$.
For performance, the median and average (geometric mean) improvement over 
baseline were $1.6X$ and $2.7X$, respectively.

\boldpara{Program size.}
For program size, Algorithm~\ref{algorithm:general_algorithm} showed improvement in 
$\numReadPhaseOneTwoThreeIncrease$ benchmarks and degradation in 
$\numReadPhaseOneTwoThreeDecrease$ benchmarks; $\numReadPhaseOneTwoThreeEqual$ benchmarks
had the same program size with both Algorithm~\ref{algorithm:general_algorithm} and \flashfill.
Again, most of the degradation cases were due to the Phase $2$ ranking model.
For program size, the median and average improvement over baseline were $1.26 X$ and
$1.53 X$, respectively, with the maximum improvement being $20 X$.

\paragraph{RQ2: Importance of each phase.}
For Phase 1, the motivation presented was that standard PBE systems are
optimized to converge to user's intent with as few examples as possible; while the same
is not possible for synthesizers tuned for other objectives.  To validate
this intuition, we compare the number of examples $e_1$ and $e_2$ required
by Phase $1$ and Phase $2$ running independently.  For performance (resp.
program size) objectives, we found that in $353$ (resp. $364$) of the
$\numBenchmarks$ cases, the number of examples required by Phase $2$ is
greater than by Phase $1$.  This justifies the use of a separate Phase $1$
to narrow down user intent.

To quantify the need for both Phases $2$ and $3$, we measure the effect of
skipping either of these phases on the objective values of the optimized
programs, i.e., we compare $o_{123}$ to $o_{13}$ (skip Phase $2$) and
$o_{12}$ with $o_{123}$ (skip Phase $3$).
Figures~\ref{fig:read:o123_vs_o13} and~\ref{fig:read:o123_vs_o12} show
ablation results for program size, and Figures~\ref{fig:perf:o123_vs_o13}
and~\ref{fig:perf:o123_vs_o12} do the same for performance.

For Phase $2$:
\begin{inparaenum}[(a)]
    \item For performance, we found that $o_{123}$ was better than $o_{13}$ in 
    $\numPerfPhaseOneTwoThreeOverPhaseOneThree$ cases, with the median and average speed-up 
    being $1.03X$ and $1.6X$, respectively.
    \item For program size, $o_{123}$ was better than $o_{13}$ in $\numReadPhaseOneTwoThreeOverPhaseOneThree$
    cases, with the median and average reduction in sizes being $1.1 X$ and $1.23 X$ respectively.
\end{inparaenum}

For Phase $3$:
\begin{inparaenum}[(a)]
    \item For performance, Phase $3$ led to an improvement in 
    $\numPerfPhaseOneTwoThreeOverPhaseOneTwo$ cases, with median
        and average speed-ups being $1.2X$ and $1.05X$, respectively.
    \item For program size, Phase $3$ led to an improvement in $\numReadPhaseOneTwoThreeOverPhaseOneTwo$
        cases, with median and average reduction in size being $1 X$ and $1.14 X$, respectively.
\end{inparaenum}

\begin{wrapfigure}[10]{R}{0.4\textwidth}
    \vspace{-1ex}
    \includegraphics[scale=0.97]{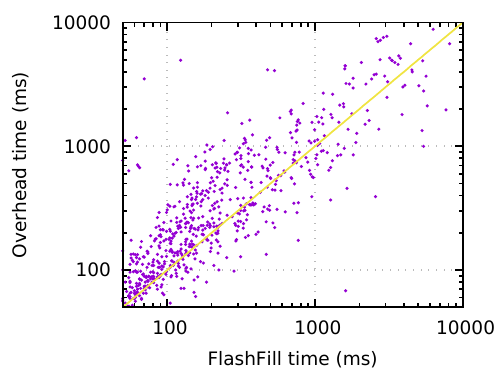}
    \vspace{-5ex}
    \caption{Algorithm~\ref{algorithm:general_algorithm} overhead}
    \label{fig:overhead}
\end{wrapfigure}

\paragraph{RQ3: Overhead of using Algorithm~\ref{algorithm:general_algorithm}.}
The overhead 
(i.e., Phase $2$ and Phase $3$ time $t_{opt}$) is plotted in
Fig.~\ref{fig:overhead} against baseline \flashfill time $t_{pbe}$.
As it can be seen, the overhead is generally close to the PBE time, with 
average and median overheads being $1.39X$ and $1.27X$.
While these numbers may seem large, the absolute overheads are quite small due to
the baseline \flashfill being quite fast.
In fact, the absolute median and average (arithmetic mean) overhead is $0.25$ and $1.23$ 
seconds, respectively, showing that efficiency of Algorithm~\ref{algorithm:general_algorithm}
is still well within the realm of responsive user interfaces where PBE is typically used.

\ignore{
Firstly, Algorithm~\ref{algorithm:general_algorithm} introduces an average and median overhead
($t_{opt} / t_{pbe}$) of $XXX$ and $YYY$ in the case of performance, and $XXX$ and $YYY$
in the case of program size (See Figure~\ref{fig:xxx}).
Further, in the case of performance, we quantify the overhead of using Algorithm~\ref{algorithm:general_algorithm}
over standard PBE in the context of the data preparation workflow shown in Figure~\ref{fig:xxx}.
Informally, a data-set $D$ justifies using Algorithm~\ref{algorithm:general_algorithm} over standard PBE
if, the time taken to process it using Algorithm~\ref{algorithm:general_algorithm} $t_{pbe} + t_{opt} +
\frac{D}{o_{123}}$ is less that the time taken to process it using standard PBE $t_{pbe}
+ \frac{D}{o_1}$.
We compute what the size of the full data-set $D$ required to justify the overhead as 
follows:
\[
    D \geq \frac{t_{opt} \cdot o_1 \cdot o_{123}}{ o_{123} - o_{1} }
\]
Figure~\ref{fig:xxx} plots the data-set size that justifies Algorithm~\ref{algorithm:general_algorithm} 
overhead.
As can be seen from the figure, the overhead is justified for typical industrial scale
data-sets.
}

The experiments provide clear evidence that each of the three phases of Algorithm~\ref{algorithm:general_algorithm}
contributes significantly in the process of solving the \qpbe problem, while incurring 
acceptable overhead altogether.


\vspace{-2ex}
\section{Related Work}
\label{section:related_work}
\vspace{-1ex}
Smith et al.~\cite{smith-19} have argued and shown experimental evidence for 
the significance of considering programs only in {\em normal} form
(i.e., those not amenable to a further rewrite using a set of 
rewrite rules), thus pruning the search space and making it more tractable.
We leverage the same observation by doing the program search over a reduced 
grammar optimized for PBE engine efficiency, but use rewrite rules in the
opposite direction to lift the result to richer target languages.

\paragraph{Quantitative Synthesis}
Synthesis with quantitative objectives has been considered in several settings.
Bloem et.al.~\cite{reactive:cav09} discuss the problem of synthesizing optimal
solutions in the context of reactive systems, where the objectives are specified
using weighted automata.
Hu and D'Antoni~\cite{qsygus:cav18} proposed a quantitative SyGuS framework that
allows for expressing quantitative constraints over the desired solution from a
weighted grammar.
Their approach is to reduce the problem to a standard SyGuS problem over a richer
non-weighted grammar that explicitly tracks weights using new non-terminals.
In contrast, we propose a quantitative PBE framework that leverages the underlying 
PBE framework's ranking engine to cater to the quantitative objective, and uses 
rewrite rules to find better solutions outside the grammar of the PBE.

Chaudhuri et.al.~\cite{smoothing:popl14} address the problem of synthesizing parameters
in a program to satisfy given boolean and quantitative objectives. Their smoothed proof
search technique reduces the problem to a series of unconstrained smooth optimization
problems that are solved using numerical methods. 
D'Antoni et.al.~\cite{qlose:cav16} address the problem of synthesizing program repairs that
meet a quantitative objective of being close to the original program in terms syntax or
execution traces. Their technique is to encode the quantitative objective as a
optimization constraint for the underlying Sketch synthesizer, which then uses an
incremental search methodology. 
In contrast, we deal with synthesizing full programs and leverage symbolic 
deductive techniques to meet the quantitative objective.

\paragraph{Synthesis of Efficient Programs}
One of the quantitative objectives that we discussed and experimented with is 
that of performance. 
Cerny et.al.~\cite{quantitative:cav11} and Vechev et.al.~\cite{concurrency:popl10} 
studied the problem of completing a partial program, by transforming and adding
synchronization constructs so that worst-case or average case performance is
optimized.
While the goal here is similar to ours, that of catering to performance criterion
during synthesis, the starting point (partial program instead of examples)
and application (concurrency) are very different.

Another classical application of program synthesis in the context of performance
has been in {\em superoptimization}, which is the task of synthesizing an optimal
sequence of instructions that is semantically equivalent to a given piece of code~\cite{superoptimization:asplos87}. 
Given the undecidability of checking semantic equivalence, superoptimization has
been restricted to optimizing straight-line code fragments~\cite{peephole:osdi08,superoptimization:asplos16}
or more generally, loop bodies~\cite{simd:ppopp13}. 
In contrast, we are able to deal with sophisticated string manipulating code 
involving complex operators by relaxing the semantic equivalence criterion to that
of equivalence under a given precondition.
Further, in \qpbe, we solve for a double optimization criterion, that of both minimizing
user interaction, as well as performance.

Various approaches to superoptimization include enumeration of instruction 
sequences~\cite{superoptimization:asplos87,superoptimization:asplos16},
reduction to SAT/SMT constraint solving~\cite{loopfree:pldi11}, and searching over constrained 
spaces of equality-preserving transformations~\cite{denali:pldi02,peephole:asplos06}.
Our rewrite rules can be seen as a relaxed (modulo inputs) version of 
semantics-preserving transformations, which can operate on rich data types
including strings and regular expressions.
Another key difference however is that we rely on a global search algorithm first,
in addition to the local rewrite rules.




Stochastic superoptimization~\cite{stochastic:cacm16} uses a two-phase approach
where the first phase finds algorithmically distinct solutions and the second phase
finds efficient implementations.
While there are similarities to our approach of finding ``DSL distinct" solutions,
followed by different ``target language" implementations, the setting and the techniques
involved are quite different.

\section{Conclusion}
\label{section:conclusion}
There is growing interest in the area of PBE, thanks to technical advances and relevant applications. 
Advanced search algorithms have enabled synthesis of programs in real time, 
while new ranking techniques have enabled synthesis from a small number of examples. 
%
However, to broaden the reach of PBE technologies, 
we need to provide the users with more control on the nature of the synthesized
program.
These programs need to be in an appropriate target language to match
a user's workflow, need to be concise/readable for easy modifiability and maintenance,
and need to be efficient to avoid computational costs on big data.
We capture these real-world requirements using the \qpbe problem. 

Our solution approach is modular and builds over advances in existing PBE systems. 
We have implemented our technique on top of an production-quality PBE system
for the domain of data transformations. 
Our experimental results show significant benefits on top of an existing
\flashfill implementation. 

{
\setlength{\itemsep}{2pt}

}

\appendix

\section{Appendix to Section~\ref{section:motivating_example}}
\label{sec:app_2}

\begin{figure}
\fontsize{8}{9}\selectfont
\begin{lstlisting}
def transform(x):
  date_string = x[15:25]
  input_date_format = r"(?<month>\d{2})/(?<day>\d{2})/(?<year>\d{4})"
  dt_obj = parse_datetime(date_string, input_date_format)
  return dt_obj.strftime("%B ") + "{0:01d}".format(dt_obj.day) +
         dt_obj.strftime(", %Y")
\end{lstlisting}
\caption{Program $P_{12}$ produced by Phase 2 on the \figref{motivating_example}.}
\label{figure:phasetwo_program}
\end{figure}

\end{document}